\documentclass[11pt]{article}

\usepackage{xcolor}
\usepackage{amsmath,amssymb,amsthm}
\usepackage[roman,full]{complexity}
\usepackage{tabularx,booktabs}
\usepackage{fullpage}
\usepackage{pifont}
\usepackage{enumitem}
\usepackage[utf8]{inputenc}

\newcommand{\keywords}[1]{\noindent\textbf{Keywords:} #1}
\newcommand{\card}[1]{\lVert #1 \rVert}

\newcommand{\cale}{\mathcal{E}}

\newcommand{\calp}{\mathcal{P}}

\newcommand{\cals}{\mathcal{S}}
\newcommand{\calt}{\mathcal{T}}
\newcommand{\calu}{\mathcal{U}}

\newcommand{\dash}{\hbox{-}\allowbreak}

\newcommand{\LANGDEF}[3]{
\begin{center}
{\small 
\begin{tabularx}{0.98\columnwidth}{ll}
\toprule
\multicolumn{2}{c}{\textsc{#1}} \\
\midrule
{\bf Given:}   & \parbox[t]{0.75\columnwidth}{#2\vspace*{1mm}}  \\
{\bf Question:}& \parbox[t]{0.75\columnwidth}{#3\vspace*{.5mm}} \\ 
\bottomrule
\end{tabularx}
}
\end{center}
}

\newcommand{\OPTDEF}[4]{
\begin{center}
{\small 
\begin{tabularx}{0.98\columnwidth}{ll}
\toprule
\multicolumn{2}{c}{\textsc{#1}} \\
\midrule
{\bf Given:}   & \parbox[t]{0.75\columnwidth}{#2\vspace*{1mm}}  \\
{\bf Solution:}& \parbox[t]{0.75\columnwidth}{#3\vspace*{.5mm}} \\ 
{\bf Measure:}& \parbox[t]{0.75\columnwidth}{#4\vspace*{.5mm}} \\ 
\bottomrule
\end{tabularx}
}
\end{center}
}

\newtheorem{theorem}{Theorem}[section]

\newtheorem{lemma}[theorem]{Lemma}
\newtheorem{proposition}[theorem]{Proposition}

\newtheorem{construction}[theorem]{Construction}

\newtheorem{definition}[theorem]{Definition}

\newcommand{\posint}{\ensuremath{\mathbb{Z}_+}}

\newcommand{\nonnegreals}{\ensuremath{\mathbb{R}_{\geq 0}}}
\newcommand{\nonnegints}{\ensuremath{\mathbb{Z}_{\geq 0}}}

\newlang{\OPT}{OPT}

\newclass{\POLYAPX}{Poly\dash\APX}
\newclass{\LOGAPX}{Log\dash\APX}
\newclass{\LOGAPXCOM}{Log\dash\APX\dash complete}

\newcommand{\strictred}{\ensuremath{\leq_{\rm strict}}}

\newcommand{\MkU}{\ensuremath{{\rm M} k {\rm U}}}

\newlang{\opt}{opt}
\newlang{\optd}{\opt\dash}

\newfunc{\sol}{sol}
\newfunc{\type}{type}

\newcommand{\myscore}{\ensuremath{{\rm score}}}

\newcommand{\condorcet}{{\rm Condorcet}}

\newcommand{\weighted}{{\rm weighted}}

\newlang{\CC}{CC}
\newlang{\DC}{DC}

\newlang{\AV}{AV}
\newlang{\DV}{DV}   

\newlang{\RPC}{RPC}
\newlang{\PC}{PC}
\newlang{\PV}{PV}

\newlang{\TE}{TE}
\newlang{\TP}{TP}

\newlang{\rpcte}{RPC\dash TE}
\newlang{\rpctp}{RPC\dash TP}
\newlang{\pcte}{PC\dash TE}
\newlang{\pctp}{PC\dash TP}
\newlang{\pvte}{PV\dash TE}
\newlang{\pvtp}{PV\dash TP}

\newlang{\UW}{UW}
\newlang{\NUW}{NUW}
\newlang{\pluralityccac}{plurality\dash CCAC}
\newlang{\pluralitydcac}{plurality\dash DCAC}
\newlang{\pluralityccdc}{plurality\dash CCDC}
\newlang{\pluralitydcdc}{plurality\dash DCDC}
\newlang{\pluralityccrpcte}{plurality\dash CCRPC \dash TE}
\newlang{\condorcetccav}{Condorcet\dash CCAV}
\newlang{\condorcetccdv}{Condorcet\dash CCDV}
\newlang{\approvalccav}{approval\dash CCAV}
\newlang{\approvalccdv}{approval\dash CCDV}
\newlang{\majorityccdc}{Majority\dash CCDC}
\newlang{\pluralityccuac}{plurality\dash CCUAC}

\newcommand{\score}[1]{\text{score}(#1)}

\newlang{\ptas}{PTAS}
\newlang{\fptas}{FPTAS}
\newlang{\npo}{NPO}
\newlang{\pluralitydcudc}{plurality\dash DCUDC}
\newlang{\condorcetccuav}{Condorcet \dash CCUAV}
\newlang{\condorcetccudv}{Condorcet \dash CCUDV}
\newlang{\pluralityccdcexclude}{plurality\dash CCDC\dash NUW \dash Exclude}
\newlang{\optpluralityccdcexclude}{opt\dash plurality\dash CCDC\dash Exclude}
\newlang{\optpluralityccrpcte} {\optd \pluralityccrpcte}
\newlang{\optpluralitydcdc}{\optd \pluralitydcdc}
\newlang{\optapprovalccav}{opt\dash approval\dash CCAV}
\newlang{\optapprovalccdv}{opt\dash approval\dash CCDV}
\newlang{\optcondorcetccav}{opt\dash Condorcet\dash CCAV}
\newlang{\optcondorcetccdv}{opt\dash Condorcet\dash CCDV}

\newlang{\optpluralitydcdcweighted}{\optd \pluralitydcdc\dash weighted}
\newlang{\optapprovalccavweighted}{opt\dash approval\dash CCAV\dash weighted}
\newlang{\optapprovalccdvweighted}{opt\dash approval\dash CCDV\dash weighted}
\newlang{\optcondorcetccavweighted}{opt\dash Condorcet\dash CCAV \dash weighted}
\newlang{\optcondorcetccdvweighted}{opt\dash Condorcet\dash CCDV \dash weighted}

\title{Approximating Electoral Control Problems}
\author{Huy~Vu Bui
\and Michael C. Chavrimootoo
\and Kien T. Le
\and Son M. Nguyen\\
Department of Computer Science\\
Denison University\\
Granville, OH 43023}
\date{August 20, 2026}

\begin{document}

\maketitle

\begin{abstract}
    Much research in electoral control---one of the most studied form of electoral attacks, in which an entity running an election alters the structure of that election to yield a preferred outcome---has focused on giving decision complexity results, e.g., membership in $\P$, $\NP$-completeness, or fixed-parameter tractability.
    Approximability
    on the other hand 
    has
    received little attention in electoral control, despite 
    its
    prevalence in the study of other forms of electoral attacks, such as manipulation and bribery. 
    Early work established 
    preliminary results 
    about
    popular voting rules such as plurality, approval, and Condorcet. 
    In this paper, we 
    completely determine
    for each of the ``standard'' control problems under plurality, approval, and Condorcet, whether they are approximable, and we prove our results in both the weighted and unweighted voter settings. 
    \medskip
\end{abstract}

\section{Introduction}

In modern societies, voting or electing is commonly used as a method to help a group reach a decision that reflects their collective preferences. 
However, during the electoral process, individuals or groups may seek to influence the outcome in their favor, for example, to make
a particular candidate win or lose the election. In particular, the setting wherein the electoral chair---the entity conducting the election---alters the structure of the election to yield a preferred outcome is known as \emph{electoral control}, which is one of the major forms of electoral manipulative attacks
(along with manipulation and bribery) within computational social choice (COMSOC); see~\cite{%
alo-ina-jai-tal-mor:c:control-liquid-democracy,
kac-rot:c:weighted-voting-games,
mau-nic-nus-rot-see:c:completing-picture-schulze-rp,
yan:c:impact-your-paper,
alo-jan-lis-pap:c:perspective-liquid-democracy,
che-kac-nus-rot-sch-see:c:control-in-comsoc,
de-dey-san:c:control-in-polls,
fal-jan-kno-pok-sch-slu-sor:c:project-strength-control,
kac-rot:c:control-by-deleting-players-from-weighted-voting-games-is-np-pp-complete-for-the-penrose-banzhaf-power-index,%
che-gut-mus-sim:c:control-hedonic-games,%
sch-sor:c:control-participatory-budgeting,%
bui-cha-le-ngu:c-ea:approximating-control,%
cha-jea:c:crc,%
cha-cli-fer-gib-hem-luu-nar-wan:c:axiomatic-tools%
} for select work on control since 2024. 
The standard control actions are adding/deleting/partitioning candidates/voters.

Much of the work on control has produced decision-complexity results (see~\cite{rot:b:econ-second-edition}), because---from a classical perspective---if a control problem is NP-hard, then it is ``hard to solve'' and so one can view this as ``protection against'' (or ``resistance to'') the attack. A wide body of work also emerged giving fixed-parameter-tractability results (see~\cite{bet-bre-che-nie:c:parameterized-elections-survey}). In practice, however, a campaign manager would prefer to compute actual solutions to problems, and that line of work, i.e., the line of search complexity, has been explored too~\cite{hem-hem-men:j:search-versus-decision,car-cha-hem-nar-tal-wel:j:search-versus-search}, sometimes showing that there is a tension between decision and search complexities. 

In this paper, we turn our focus to the optimization versions of the ``standard'' electoral control problems, where the chair seeks to minimize some property of the solution to the control problem, e.g., minimizing the number of candidates to delete. To our surprise, this line of work is well-studied with respect to manipulation and bribery, but not control;
see \cite{bre-fal-hem-sch-sch:c:approximability-of-manipulation,fal-hem-hem-rot:j:llull,fal-hem-hem:j:bribery,xia:margin-of-victory,yan:c:election-few-candidates,kel-has-haz:j:approximating,kel-has-haz:j:approximate-weighted-priced-bribery}
to list a few examples.
In particular, Brelsford~\cite{bre:msthesis:approximations} initiated that study by providing preliminary results and by outlining difficulties one would likely face in designing approximation algorithms (or proving the lack thereof) for the standard control problems. Eventually, Faliszewski, Hemaspaandra, and Hemaspaandra~\cite{fal-hem-hem:j:weighted-control} and  Bredereck et al.~\cite{bre-fal-nie-sko-tal:j:mixinteger-lp} contributed to that list of result. Thus the amount of work that has been done over the last 19 years has been very limited.

Our work complements the view that $\NP$-hardness does not necessarily equate to ``protection'' as many $\NP$-hard problems admit efficient approximations (for example, the Vertex Cover problem admits a 2-approximation \cite{cor-lei-riv-ste:b:algorithms-fourth-edition}) and in practice, a suboptimal solution may be perfectly acceptable. 
So there's a very real need (if one wishes to treat ``resistance'' as a form of protection) to argue that approximation algorithms do not undermine the resistance captured by the NP-hardness of the various electoral control problems. 

This work is not antithetical to the aforementioned work on fixed-parameter tractability, but is rather a different way of analyzing \NP-hard control problems, and can also offer additional applications as some FPT algorithms lose their appeal even for relatively small parameter values. For example, an algorithm that runs in $O(2^kn^3)$, where $n$ is the size of the input and $k$ is the parameter may not be practical when $k=100$. 
Moreover, in this paper, we give approximation algorithms for problems that do not admit FPT algorithms~\cite{bet-uhl:j:fpt-voting-graph,fen-lua-liu-zhu:j:fpt-voting}, which is interesting.

\paragraph{Contributions} 
We give new approximation algorithms, hardness/completeness results, and a refined classification of resistance under plurality,  approval, and Condorcet through the lens of approximation.
We summarize our results in Table~\ref{table:results}. 
We show that partition-based control problems are inapproximable 
regardless of the ``measure'' being optimized, and our proof uses a simple argument that successfully bypasses difficulties anticipated by~\cite{bre:msthesis:approximations}.
Finally we prove the inapproximability of the remaining $\NP$-hard control problems.
In addition to these classifications, we provide matching results for each of the weighted variants of the standard control problems. 

\newcommand{\notapprox}{Inapprox.}

\label{table:results}
\begin{table*}[ht]
\centering
\scriptsize

\begin{tabular}{|l||cc||cc||cc|}
\hline
{} & \multicolumn{2}{c||}{Plurality} & \multicolumn{2}{c||}{Condorcet} & \multicolumn{2}{c|}{Approval} \\
\cline{2-7}
Control by & CC & DC & CC & DC & CC & DC \\
\hline
\hline
Adding Candidates & \notapprox   & \notapprox   & - & - & - & - \\ 
\hline
Deleting Candidates & \textbf{Poly-APX}/{\boldmath$\Omega(m^{1/4})$} & \textbf{\notapprox} & - & - & - & - \\ 
\hline
\begin{tabular}[c]{@{}l@{}} Partition of\\Candidates\end{tabular} & 
  \begin{tabular}[c]{@{}c@{}}TE: \textbf{\notapprox} \\ TP: \textbf{\notapprox}\end{tabular} & 
  \begin{tabular}[c]{@{}c@{}}TE: \textbf{\notapprox} \\ TP: \textbf{\notapprox}\end{tabular} & 
  - & - &
  \begin{tabular}[c]{@{}c@{}}TE: - \\ TP: -\end{tabular} & 
  \begin{tabular}[c]{@{}c@{}}TE: - \\ TP: -\end{tabular}
  \\
\hline
  \begin{tabular}[c]{@{}l@{}} Run-off Partition\\ of Candidates\end{tabular} & 
  \begin{tabular}[c]{@{}c@{}}TE: \textbf{\notapprox} \\ TP: \textbf{\notapprox}\end{tabular} & 
  \begin{tabular}[c]{@{}c@{}}TE: \textbf{\notapprox} \\ TP: \textbf{\notapprox}\end{tabular} & 
  - & - &
  \begin{tabular}[c]{@{}c@{}}TE: - \\ TP: -\end{tabular} & 
  \begin{tabular}[c]{@{}c@{}}TE: - \\ TP: -\end{tabular} 
  \\
\hline
Adding Voters & - & -  & \textbf{\notapprox} & - & \textbf{Log-APX-Complete} & -\\ 
\hline
Deleting Voters & - & - & \textbf{\notapprox} & - & \textbf{Log-APX-Complete} & - \\ 
\hline
Partition of Voters & 
  \begin{tabular}[c]{@{}c@{}}TE: - \\ TP: \textbf{\notapprox}\end{tabular} & 
  \begin{tabular}[c]{@{}c@{}}TE: -\\ TP: \textbf{\notapprox}\end{tabular} &  
  \textbf{\notapprox} & - &
  \begin{tabular}[c]{@{}c@{}}TE: \textbf{\notapprox} \\ TP: \textbf{\notapprox}\end{tabular} & 
  \begin{tabular}[c]{@{}c@{}}TE: - \\ TP: -\end{tabular}
  \\
\hline
\end{tabular}

\caption{Approximation complexities of the $\NP$-hard standard control problems under plurality, Condorcet, and approval. Our (new) results are boldfaced. Nonboldfaced results are from 
\cite{bre:msthesis:approximations}. A ``-'' indicates that the decision problem is in $\P$.}
\end{table*}

\section{Preliminaries}

In this section, we will  first cover preliminaries on election systems,  control types, and related definitions. These follow the same foundational definitions established by Bartholdi, Tovey, and Trick~\cite{bar-tov-tri:j:control}, which were later expanded in the work of Hemaspaandra, Hemaspaandra, and Rothe~\cite{hem-hem-rot:j:destructive-control} and many other papers in the field (see for example, \cite{bra-con-end-lan-pro:b:handbook-of-comsoc,rot:b:econ-second-edition}). We then introduce the necessary terminology and concepts related to approximation algorithms.

We assume standard familiarity with the complexity classes $\P$ and $\NP$, as well as with the notions of $\NP$-completeness (with respect to polynomial-time many-one reductions,\footnote{We say that $A$ polynomial-time many-one reduces to $B$ ($A \leq_m^p B$) if there is a polynomial-time computable function $f$ such that for each $x$, $x \in A \iff f(x) \in B$.} as is standard) and fixed-parameter tractability.

For each $\mathbb{F} \in \{\mathbb{Z}, \mathbb{R}\}$, we let 
$\mathbb{F}_+ = \{x \in \mathbb{F} \mid x > 0\}$ and let $\mathbb{F}_{\geq 0} = \{x \in \mathbb{F} \mid x \geq 0\}$. For each $n \in \posint$ we adopt the shorthand $[n] = \{1, \ldots, n\}$. 
We also adopt the 
(usual) 
convention that $\max\emptyset = -\infty$ and $\min\emptyset = \infty$.\footnote{We abuse notation in this paper by using $\infty$ (resp. $-\infty$) in an algorithmic context to refer to a special symbol that is strictly larger (resp.\ smaller) than every number used by the algorithm(s).}

\subsection{Elections and Election Systems}

An \emph{election} is a pair of two finite sets: a  candidate set $C$ and a collection of votes~$V$, where each vote is implicitly associated to a voter.
In this paper, each vote is either a linear ordering 
over the set of candidates, e.g., if $C = \{a, b, c\}$ then $b > c > a$ is a vote, or an approval ballot, i.e., a $\card{C}$-bit string where each bit corresponds to a candidate, and a bit is set to 1 exactly if the corresponding candidate is approved (otherwise, that bit is set to 0) in that vote.
Whenever we use ``$\cdots$'' in a linear ordering, this represents the remaining candidates in some predefined way, e.g., ordered by name.

A voting rule (aka., election system)
$\cale$ is defined as a function that maps each election $(C, V)$ to a subset of $C$, representing the set of winning candidates. In this paper, will study more closely the three voting rules defined below.

\begin{itemize}%
\item \textbf{Plurality:} In a plurality election, each vote is a linear order. Under plurality, each candidate receives one point for each vote that ranks them first, and a plurality winner is a candidate with maximal points.
If multiple such candidates exist, they are all winners.

\item \textbf{Approval:} In an approval election, each vote is an approval ballot.
Under approval, each candidate receives one point for each vote that approves them, and an approval winner is a candidate with maximal points. If multiple such candidates exist, they are all winners.

\item \textbf{Condorcet:} In a Condorcet election, each vote is a linear order.
A Condorcet winner of an election $(C, V)$ is a candidate $c$ who defeats each candidate $d \neq c$
in a head-to-head contest, i.e., more than half of the votes rank $c$ above $d$.
If a Condorcet winner exists, then that candidate is a unique winner. Otherwise, no one wins.
\end{itemize}

\subsection{Electoral Control Problems}

In their seminal work, Bartholdi, Tovey, and Trick~\cite{bar-tov-tri:j:control} defined the notion of ``constructive control'' in the ``unique-winner model,'' and their work was substantially expanded upon in subsequent work.

In this section, we define the
``standard
control problems''---which were originally defined by \cite{bar-tov-tri:j:control} and \cite{hem-hem-rot:j:destructive-control}---%
along with certain natural variants of those problems that have been raised in the literature and for which we are able to provide results. Each of these problems has been originally defined in the literature as a decision problem, but our work is about approximation algorithms, so we need to provide definitions for the optimization versions of these problems. 

Each partition-based problem in Definition~\ref{d:control-types} consists of subelections whose winners proceed to a final round. To remain consistent with the work of \cite{bre:msthesis:approximations}, we use as tie-breaking rules for subelections the ties-eliminate (TE) rule---in which only a unique winner of a subelection proceeds to the next round---and the ties-promote (TP) rule---in which all winners of a subelection proceed to the next round.  For each such problem, every election (in each stage) is conducted using the same voting rule. Thus each partition-based-problem definition implicitly defines two problems, and so Definition~\ref{d:control-types} specifies 10 decision control problems for each voting rule~$\cale$. 

Related to partitioning, we define a partition of a candidate set $C$ as a pair of disjoint subsets $C_1$ and $C_2$ such that $C_1 \cup C_2 = C$.
Similarly, we define a partition of a collection of votes $V$ as a pair $V_1$ and $V_2$ that when merged (denoted by $V_1 \cup V_2$, as is consistent with the prevalent notation in the literature) equals $V$.

\begin{definition}[\cite{car-cha-hem-nar-tal-wel:j:sct}]\label{d:control-types}  
Let $\cale$ be a voting rule.
\begin{enumerate}
    \item\label{i:ccac} In the \textbf{constructive control by adding candidates} problem for $\cale$ (denoted by $\cale\dash\CC\AC$), we are given two disjoint sets of candidates $C$ and $A$, a collection of votes $V$ over $C \cup A$, a candidate $p \in C$, and a nonnegative integer $k$. We ask if there is a set $A' \subseteq A$ such that (i)~$\|A'\| \leq k$ and (ii)~$p$ uniquely wins $\cale$ election $(C\cup A', V)$.

    \item\label{i:ccdc} In the \textbf{constructive control by deleting candidates} problem for $\cale$ (denoted by $\cale\dash\CC\DC$), we are given an election $(C, V)$, a candidate $p \in C$, and a nonnegative integer $k$. We ask if there is a set $C' \subseteq C-\{p\}$ such that (i)~$\|C'\| \leq k$ and (ii)~$p$ uniquely wins $\cale$ election $(C - C', V)$.
    
    \item\label{i:ccav} In the \textbf{constructive control by adding voters} problem for $\cale$ (denoted by $\cale\dash\CC\AV$), we are given a set of candidates $C$, two collections of votes, $V$ and $W$ (respectively representing the ``registered'' and ``unregistered'' voters), over $C$, a candidate $p \in C$, and a nonnegative integer $k$. We ask if there is a collection $W' \subseteq W$ such that (i)~$\|W'\| \leq k$ and (ii)~$p$ uniquely wins $\cale$ election $(C, V \cup W')$.
    
    \item\label{i:ccdv} In the \textbf{constructive control by deleting voters} problem for $\cale$ (denoted by $\cale\dash\CC\DV$), we are given an election $(C, V)$, a candidate $p \in C$, and a nonnegative integer $k$. We ask if there is a collection $V' \subseteq V$ such that (i)~$\|V'\| \leq k$ and (ii)~$p$ uniquely wins $\cale$ election $(C, V-V')$.

        \item In the \textbf{constructive control by partition of voters} problem for $\cale$, in the TP or TE tie-handling rule model,
    (denoted by $\cale\dash\CC\PV\dash\TP$ or $\cale\dash\CC\PV\dash\TE$, respectively),  
    we are given an election $(C, V)$, and a candidate $p \in C$. We ask if there is a partition $(V_1, V_2)$ of $V$
    such that $p$ uniquely wins $\cale$ election $(W_1 \cup W_2, V)$,
    where $W_i$ denotes the winners of $\cale$ subelection $(C, V_i)$ that survive the tie-handling rule.

    \item In the \textbf{constructive control by run-off partition of candidates} problem for $\cale$, in the TP or TE tie-handling rule model,
    (denoted by $\cale\dash\CC\RPC\dash\TP$ or $\cale\dash\CC\RPC\dash\TE$, respectively), 
    we are given an election $(C, V)$, and a candidate $p \in C$. We ask if there is a partition $(C_1, C_2)$ of $C$  such that $p$ uniquely wins
    $\cale$ election $(W_1 \cup W_2, V)$, where $W_i$ denotes the winners of $\cale$ subelection $(C_i, V)$ that survive the tie-handling rule.

    \item In the \textbf{constructive control by partition of candidates} problem for $\cale$, in the TP or TE tie-handling rule model,
    (denoted by  $\cale\dash\CC\PC\dash\TP$ or $\cale\dash\CC\PC\dash\TE$, respectively),  
    we are given an election $(C, V)$, and a candidate $p \in C$. We ask if there is a partition $(C_1, C_2)$ of $C$  such that $p$ uniquely wins
    $\cale$ election $(W_1 \cup C_2, V)$, where $W_1$ denotes the winners of $\cale$ subelection $(C_1, V)$ that survive the tie-handling rule.

    \end{enumerate}
\end{definition}

For each control type defined above, several variants of the control problems can be obtained by modifying certain aspects of the problem definition as follows.

\begin{description}
    \item[Destructive Control:] 
    Definition \ref{d:control-types} gives the constructive setting, in which the goal is to ensure that a designated candidate $p$ is a unique winner of the election. In the \emph{destructive} variant, the objective is to \emph{prevent} $p$ from being a unique winner. To indicate this variant, the prefix ``CC'' is replaced with ``DC'' in the problem name.

    \item[Weighted Votes:] In the weighted electoral control setting, the input also contains a function mapping each vote $v$ to a positive integer $w(v)$. 
    Under plurality each candidate receives $w(v)$ points for each vote $v$ that ranks them first, under approval each candidate receives $w(v)$ points for each vote $v$ that approves them, and under 
    Condorcet each pairwise contest can be interpreted as a weighted plurality election.
    The problems in Definition~\ref{d:control-types} are all unweighted, and in our paper, we will give results for both weighted and unweighted control problems. To refer to the weighted version of a control problem, we add the suffix ``-weighted.''
\end{description}

To our knowledge, the earliest comprehensive study of weighted electoral control was done by Faliszewski, Hemaspaandra, and Hemaspaandra~\cite{fal-hem-hem:j:weighted-control}, even though several earlier theses (e.g., \cite{bre:msthesis:approximations,rus:msthesis:borda,lin:thesis:elections}) had been discussing the notion.

Let $\cale$ be a voting rule and let $\calt$ be a control action (e.g., $\CC\AC$). If $\cale\dash\calt$ is $\NP$-hard, we say that $\cale$ is resistant to $\calt$.

We now turn our attention to the optimization versions of the above control problems. First, we define $\NPO$, a class of optimization problems. Informally, an NPO problem is an optimization problem whose decision version is in $\NP$\@. While we could state our definitions without referring to NPO, we have opted to include this definition to provide the reader with more concrete background on optimization problems.

\begin{definition}[\protect\cite{cre:c:approximation}]\label{d:npo}
    An $\NPO$ problem $A$ is a four-tuple $\allowbreak (I_A, \sol_A, \mu_A, \type_A)$ such that
    \begin{enumerate}
        \item $I_A$ is the set of instances of $A$ and it is recognizable in polynomial time.
        \item Given an instance $x \in I_A$, $\sol_A(x)$ denotes the set of feasible solutions of $x$, and there is a polynomial $p$ such that for each $y \in \sol_A(x)$, $|y| \leq p(|x|)$. Moreover, for each $x$ and $y$ with $|y| \leq p(|x|)$, it is decidable in polynomial time whether $y \in \sol_A(x)$.
        \item Given an instance $x \in I_A$ and an arbitrary $y$, if $y \in \sol_A(x)$, then function $\mu_A(x, y)$ denotes the positive integer measure of $y$ (often also called the value of $y$). 
        If $y \not\in \sol_A(x)$, then $\mu_A(x, y) = \type_A \emptyset$.
        The function $\mu_A$ is computable in polynomial time.
        \item $\type_A \in \{\min, \max\}$.
    \end{enumerate}

     Finally, for each $x \in I_A$, let $\OPT_A(x) = y$, where $|y|$ is polynomial in $|x|$, such that $$\mu_A(x, y) = \type_A\{\mu_A(x, y') \mid y' \in \sol_A(x)\}.$$
\end{definition}

Defining the optimization versions of the nonpartition-based control problems is very natural, and we define them by way of example. All such problems are minimization problems.
We prepend our optimization problems with ``\optd'' to differentiate between the optimization and decision versions of a problem.
The optimization problem for $\cale\dash\CC\AC$, i.e., $\optd\cale\dash\CC\AC$, is defined as follows.

\OPTDEF{Minimum Constructive Control by Adding Candidates under $\cale$ ($\optd\cale\dash\CC\AC$)}
{Two finite sets $C$ and $A$ of candidates, and collection $V$ of votes over $C \cup A$, and a distinguished candidate $p \in C$.}
{A set $A' \subseteq A$ such that $\{p\} = \cale(C\cup A', V)$.}
{$1 + \|A'\|$.\footnote{We add one to ensure that our measure is positive, which will prevent us from running into issues of division by zero when computing ``performance ratios'' (see Definition~\ref{def:performance}). That does not affect how we compute such ratios in this paper as we always are providing asymptotic bounds on performance ratios.}}

The optimization versions of the remaining nonpartition-based control problems are defined analogously. In particular, they are all minimization problems, and
when adding/deleting candidates/votes, a feasible solution $S$ is a collection of candidates/votes to be added/deleted to make $p$ a unique winner, and the measure is $1 + \card{S}$.

Defining optimization problems for the partition-based problems is not straightforward; although it is clear what the inputs and feasible solutions should be, it is not evident what the measures should be. Brelsford~\cite{bre:msthesis:approximations} proposes maximizing the size of the smallest part of the partition---thereby minimizing the size difference between the parts of the partition---%
which is
consistent with recent work~\cite{cha-erd-hem-hem-kot-nie-rei-smi:manuscript:npm} that builds on the work of Erd\'{e}lyi, Hemaspaandra, and Hemaspaandra~\cite{erd-hem-hem:c:more-natural-models-of-partition-control}, which generalizes the notion of partitions to bound the difference between the sizes of the parts. However, we do not fix a particular measure for the partition-based problems in this paper as our inapproximability results for partition-based control problems will hold \emph{for all measures}.

\subsection{Approximations and Hardness}

Approximation complexity classes provide a natural framework for describing and comparing the performance of algorithms that seek near-optimal solutions. In this work, we analyze the approximability of several electoral control problems, and our results are framed in terms of these classes. To make the statements of our theorems precise, we define the key approximation classes that will appear in this paper. 

We present the following definitions, taken from Crescenzi's \cite{cre:c:approximation} survey.

\begin{definition}[\protect\cite{cre:c:approximation}]\label{def:performance}
Let $A = (I_A, \sol_A, \mu_A, \type_A)$ be an optimization problem, let $x \in I_A$, and let $y \in \sol_A(x)$. The \emph{performance ratio} of $y$ with respect to $x$ is defined as
$$R_A(x, y) = \max\left\{ \frac{\mu_{A}(x, y)}{\OPT_A(x)}, \frac{\OPT_A(I)}{\mu_A(x, y)} \right\}.$$
\end{definition}

Throughout the literature, the terms ``performance ratio'' and ``approximation ratio'' are used interchangeably.  

An $f(n)$-approximation algorithm for an $\NPO$ problem $A$ is a polynomial-time computable function $M$ that on each input $x \in I_A$ outputs a solution $y$ such that $R_A(x, y) \leq f(|x|)$, and we at times in an abuse of terminology refer to $f(n)$ as the performance/approximation ratio of the approximation algorithm (with respect to problem $A$)\@.
We assume that if there is no such solution $y$, then the algorithm outputs no solution.
If a problem has/admits no $f(n)$-approximation algorithm, we say it is inapproximable or not approximable.
We now move on to the subclasses of $\NPO$ that are relevant to our work.\footnote{%
We mention in passing that Brelsford~\cite{bre:msthesis:approximations} also considered three subclasses of $\NPO$ known as $\FPTAS$, $\PTAS$, and $\APX$ (which we do not define as we do not prove any results directly related to them), and asked whether any of the control problems they studied could be in those classes.
However, it is known that $$\FPTAS \subseteq \PTAS \subseteq \APX \subseteq \LOGAPX \subseteq \POLYAPX \subseteq \NPO,$$ and that those inclusions are strict if and only $\P \neq \NP$~\cite{cre:c:approximation}. By the nature of our results (which will be clear in Sections~\ref{s:approx} and~\ref{s:inapprox}), it follows that the answer to the above question is ``no,'' unless $\P=\NP$\@.%
}
Let $A \in \NPO$.

\begin{itemize} 
    \item $A \in \LOGAPX$ if $A$ admits an $O(\log{n})$-approximation algorithm.
    \item $A \in \POLYAPX$ if for some $k > 0$, $A$ admits an $O(n^k)$-approximation algorithm.
\end{itemize}

Just as $\NP$-completeness allows us to identify problems that do not have polynomial-time algorithms (unless $\P=\NP$), we can prove the ``optimality'' of our approximation algorithms by giving completeness results for the corresponding (approximation) class. 
However, the notion of reduction used here is not the same as the one used for $\NP$-completeness. Indeed, $\APX$-completeness is defined in terms of ``PTAS reductions'' while $\LOGAPX$-completeness and $\POLYAPX$-completeness are defined in terms of  ``AP reductions.'' It is also common to prove lower bound results using ``L reductions.'' 
All of these reductions are known as ``approximation-preserving reductions.''
Thankfully, our work has the nice property that every approximation-preserving reduction we provide is a ``strict reductions'' (defined below), and since it is known that a strict reduction from $A$ to be $B$ implies the existence of an L/PTAS/AP reduction from $A$ to $B$ (see~\cite{cre:c:approximation}), we do not define those three types of reductions; interested readers can consult~\cite{cre:c:approximation}.

\begin{definition}[Strict Reduction~\protect\cite{cre:c:approximation}]
    Let $A = (I_A, \sol_A, \mu_A, \type_A)$ and $B = (I_B, \sol_B, \mu_B, \type_B)$ be two \NPO\ problems. Then $A$ strict reduces to $B$ ($A \strictred B$) if there are polynomial-time computable function $f$ and $g$ such that for each $x \in I_A$ and each $y \in \sol_B(f(x))$, 
    $$R_A(x, g(x, y)) \leq R_B(f(x), y).$$
\end{definition}

Since strict reductions are transitive, to prove that a problem $A$ is Log-APX-complete in this paper, we first prove that it is in Log-APX, and then give a strict reduction from a known Log-APX-complete problem to~$A$\@. We do not prove any Poly-APX-completeness results in this paper.

\section{Approximations}\label{s:approx}
For approval, we prove that the optimization problems of constructive control by adding/deleting voters are \LOGAPXCOM. Membership in \LOGAPX\ is obtained through a reduction to Covering Integer Programming (CIP). For constructive control by deleting candidates under plurality voting, we give a simple $O(m)$ approximation and establish an $\Omega(m^{1/4})$ conditional lower bound via a strict reduction from Minimum~$k$-Union.  To our knowledge, this is the first application of $\MkU$ in COMSOC.

\subsection{Approval and Voter Control}

To prove our $\LOGAPX$-membership results, we appeal to the notion of Covering Integer Programming (CIP), which generalizes classic covering problems such as Set Cover and Weighted Multiset Cover by allowing for additional constraints and integer variables.
A CIP is itself a special case of an Integer Linear Program (ILP).

\begin{definition}[\protect\cite{kol-you:j:approx-cip}]
Given $A \in \nonnegreals^{m \times n}$, $a \in \nonnegreals^m$, and $c,d \in \nonnegreals^n$, a Covering Integer Program (CIP) $\mathcal{P} = (A, a, c, d)$ seeks to minimize $c^T x$ subject to $x \in \nonnegints^n$, $x \leq d$, and $Ax \geq a$.
\end{definition}

In the above program, the inequalities are so-called component-wise inequalities and ``$A x \geq a$'' is called the covering constraint.
To prove that $\approvalccav\dash\weighted$ and $\approvalccdv\dash\weighted$ are in $\LOGAPX$-complete, we leverage a result from Kolliopoulos and Young~\cite{kol-you:j:approx-cip} showing that CIPs can be approximated within a factor of $O(\ln m)$.%
\footnote{The fact that both problems have the same complexity feels natural---and perhaps even obvious---as the two problems seem to be complementary: added votes that seem to correspond to votes that should not be deleted, and deleted votes that seem to correspond to votes that that should not be added. However that relationship does not provide any insights from a complexity perspective as there is no mechanism to prevent the deletion of registered votes in the above formulation. More concretely, this is exemplified by the fact that under 3-approval, ${\rm CCAV}$ is in $\P$ while ${\rm CCDV}$ is $\NP$-complete~\cite{lin:thesis:elections}.}

We first focus on constructive control by adding voters.
For the sake of clarity, we state the \LOGAPX-membership result as a separate proposition and subsequently invoke it to prove our \LOGAPX-completeness result.

\begin{proposition}\label{prop:opt-approval-ccav-w-approx}
    Both
    $\optapprovalccav$ and $\optapprovalccavweighted$ are in \LOGAPX.
\end{proposition}
\begin{proof}
    Let $\calt = \optapprovalccavweighted$ and let $I = (C, V,\allowbreak W, w, p)$ be an instance of $\calt$.
    We assume without loss of generality that every vote in $W$ approves $p$ because adding votes that do not approve $p$ would only raise the score of other candidates rather than that of $p$. We also assume without loss of generality that no vote in $W$ approves every candidate as such a vote can be excluded from the vote addition without changing the outcome of the election.

    For a candidate $x$, let $V_x$ be the set of votes that approve $x$. We define the score of $x$ in the election as
    $\score{x} = \sum_{v \in V_x} w(v).$
    If $p$ is already a unique approval winner, we need not add any votes. 
     Since each additional vote increases $p$’s score by that vote's weight and increases that of any other candidate by the same value, any candidate whose initial score is strictly less than that of $p$ can never overtake~$p$.

    Let $C'$ denote the set of candidates whose score is at least that of $p$. For each $c \in C'$ define
    \begin{itemize}
        \item $B_c = \{ v \in W \mid v \text{ approves } p \text{ but not } c\}$ and
        \item $\Delta_c = \score{c} - \score{p}.$
    \end{itemize}
    
    For every $c \in C'$, we must add votes drawn from $B_c$ whose total contribution to $p$’s score is at least $\Delta_c+1$, thereby ensuring that $p$ finishes with strictly more approvals than every other candidate. It follows that
    \[
      W  = \bigcup_{c \in C'} B_c.
    \]
    For every vote $v \in W$, define the binary variable
    \[
      x_v =
      \begin{cases}
        1 & \text{if vote } v \text{ is added},\\
        0 & \text{otherwise}.
      \end{cases}
    \]
    We now obtain the minimum number of added votes required to make $p$ a unique
    approval winner by solving the following integer linear program:
    \[
    \begin{aligned}
    \text{Minimize} & \sum_{v\in W} x_v \\
    \text{subject to} \\
      &\sum_{v\in B_c} w(v)x_v \ge \Delta_c + 1
      && \text{for each } c\in C',\\
    & x_v \in \{0,1\}
      && \text{for each } v\in W.
    \end{aligned}
    \]

    Let us now argue the correctness of our algorithm.
    
    \begin{lemma}\label{lemma:cip-iff-control}
        Let $I=(C,V,W,w,p)$ be an instance of \optapprovalccavweighted\ and let $\calp = (A, b, c, d)$ be the covering integer program defined above. There exists a feasible integer solution to $\calp$ of objective value $k$ if and only if there exists a subset $W'\subseteq W$ with $\|W'\| = k$ such that $p$ is a unique approval winner of $(C, V \cup W')$. %
    \end{lemma}
    \begin{proof}
        For the left to right direction, let $\{x_v^*\}_{v\in W}$ be a feasible solution of $\calp$ with $\sum_{v\in W}x_v^*=k$ and choose $W'=\{v\in W\mid x_v^*=1\}$.
        For any $c\in C'$ we have
        \[
          \score{p}+\sum_{v\in B_c}w(v)x_v^*\ge\score{p}+\Delta_c+1>\score{c},
        \]
        while the scores of candidates $C-C'$ are already strictly below that of $p$. Hence, after adding votes $W'$, $p$ is a unique winner. The number of added votes equals the CIP objective.

        For the right to left direction, suppose $W'\subseteq W$ with $\|W'\|=k$ makes $p$ a unique winner when added. Necessarily $W'\subseteq B$, because adding a voter who disapproves $p$ could not help $p$ win.
        For every $c\in C'$ we must have
        \[
          \!\!\sum_{v\in W'\cap B_c}w(v)\ge\Delta_c+1,
        \]
        Defining $x_v=1$ if $v\in W'$ and 0 otherwise yields a feasible integer solution of
        $\calp$ whose objective value is $k$.
        
        Thus $\calp$ has a solution of cost $k$ if and only if there exists a subset $W'\subseteq W$ with $\|W'\| = k$ such that $p$ is a unique approval winner of $(C, V \cup W')$.
    \end{proof}

    The resulting integer linear program is an instance of CIP which admits an $O(\ln m)$-approximation algorithm (and this factor is optimal unless $\P = \NP$)~\cite{kol-you:j:approx-cip}. By Lemma~\ref{lemma:cip-iff-control}, solving this CIP will give the set of added votes for $\optapprovalccavweighted$ thus showing the problem in \LOGAPX.

    Since the \optapprovalccav\ is a special case of the \optapprovalccavweighted\ by setting the output of the weight function $w$ to always be one, we have that \optapprovalccav\ is also in \LOGAPX.
\end{proof}

Having demonstrated that \optapprovalccav\ and \optapprovalccavweighted\ belong to \LOGAPX, we now prove their \LOGAPX-completeness through a strict reduction from the Log-APX-complete problem Minimum Set Cover.

\OPTDEF{Minimum Set Cover (MSC)}
{A finite set $\calu = \{ x_1, x_2,\dots, x_n \}$ of $n$ elements and a family $\mathcal{S} = \{S_1, S_2,\dots, S_m\}$ of subsets $S_i$ of $\calu$, where the union of all sets in $\cals$ covers $\calu$.}
{A subfamily $\cals' \subseteq \cals$ such that $\calu = \bigcup_{S \in \cals'} S$.}
{The value $\|\cals'\|$.}

\begin{theorem}
    Both \optapprovalccav\ and \optapprovalccavweighted\ are \LOGAPXCOM.
\end{theorem}
\begin{proof}
    It is clear that \optapprovalccav\ strict reduces to {\rm opt-approval-CCAV-weighted} and that both problems are in $\LOGAPX$, so it suffices to prove that MSC strict reduces to \optapprovalccav.
    For brevity, let $\calt = \optapprovalccav$.

    Let $f$ be a function that given an MSC instance $(\calu,\cals)$, where $\calu = \{ x_1,\dots, x_n \}$, $\mathcal{S} = \{S_1,\dots, S_m\}$, and $S_i \subseteq \calu$ for each $i \in [m]$, outputs a $\calt$ instance $(C, V, W, p)$ such that
    \begin{itemize}
        \item the candidate set is $C = \calu \cup \{p\}$, where $p$ is the distinguished candidate,
        \item the set of registered votes is $V = \emptyset$, and 
        \item the set of unregistered votes $W$ consists of $m$ votes: For each $i\in [m]$, there is one vote $w_i$ in $W$ that only approves the candidates in $C-S_i$.
    \end{itemize}

    Let $g$ be a function that given an instance $I = (\calu, \cals)$ of MSC and a solution $W'$ of $f(I) = (C, V, W, p)$ outputs $\widehat{\cals} = \{S_i \mid w_i \in W'\} \subseteq S$ as a solution to instance $I$ w.r.t MSC.

    Note that both $f$ and $g$ are polynomial-time computable. We now prove the correctness of our reduction.

       We now prove that $\OPT_{\calt}(f(I)) = \OPT_{{\rm MSC}}(I)$. 
    Let $\cals^*$ denotes an optimal solution to MSC on input $I$ and $W^*$ denotes an optimal solution to $\calt$ on input $f(I)$.

    First, we show that $\OPT_{\calt}(f(I)) \ge \OPT_{{\rm MSC}}(I)$. 
    If $p$ can be made the unique approval winner, then for every candidate $c \in C - \{p\}$ we must add at least one vote that approves $p$ but not $c$, so that $p$’s approval count exceeds that of $c$. Because such a vote is added for each $c \in C - \{p\}$, the selected votes correspond to a solution of size $\card{W^*}$ to $I$, i.e., $\OPT_{\calt}(f(I)) = \card{W^*} \geq \OPT_{{\rm MSC}}(I)$.

    Next, we show that  $\OPT_{\calt}(f(I)) \le \OPT_{{\rm MSC}}(I)$. 
    Consider the votes $W_{\cals^*} = \{w_i \mid S_i \in \cals^*\}$. Fix any candidate $x_j \in \calu$. Since $\cals^*$ covers $\calu$, there exists a $S_{i}\in \cals^*$ containing $x_j$, and the corresponding vote $w_{i}$ disapproves $x_j$ and approves $p$. Since all added votes approve $p$, $p$ is a unique approval winner. Thus we have produced a solution of size $\card{\cals^*}$ to $f(I)$, i.e., $\OPT_{{\rm MSC}}(I) = \card{\cals^*} \geq \OPT_{\calt}(f(I))$.

    Finally, let $W'$ be an arbitrary solution to $\calt$ on input $f(I)$. We now prove that $$R_{{\rm MSC}}(I, g(I, W')) \leq R_\calt(f(I), W').$$ Let $\cals' = g(I, W')$. By definition and the fact that $\OPT_{\calt}(f(I)) = \OPT_{{\rm MSC}}(I)$, it suffices to show that $\card{\cals'} \leq \card{W'}$. In the computation of $g$, we have $\card{\widehat{\cals}} = \card{W'}$, so $\card{\cals'} = \card{W'}$. Therefore $\card{\cals'} \leq \card{W'}$.

    Thus MSC strict reduces to $\optapprovalccav$.
\end{proof}

We now prove that \optapprovalccdv\ and \optapprovalccdvweighted\ are \LOGAPX-complete
by using a similar approach to the one just used.

\begin{proposition}\label{prop:opt-approval-ccdv-w-approx}
    Both \optapprovalccdv\ and \optapprovalccdvweighted\ are in \LOGAPX.
\end{proposition}
\begin{proof}
    Let $\calt = \optapprovalccdvweighted$ , $I = (C, V, w,  p)$ be an instance of $\calt$.
    For a candidate $x$, let $V_x$ be the set of votes that approve $x$. We define the score of $x$ in the election
    \[
    \score{x} = \sum_{v \in V_x} w(v).
    \]
    If $p$ is already the unique approval winner, we need not delete any votes that.  Otherwise, we delete only votes that do not approve $p$, because deleting votes that approve $p$ would lower the score of $p$ rather than that of its opponents.  Since each deleted
    voter reduces the score of every candidate he approves, while $p$’s approval count remains unchanged, any candidate whose initial score is strictly less than the score of $p$ can never overtake $p$.

    Let $C'$ denote the set of candidates whose score is at least the score of $p$.  For each $c \in C'$ define
    \begin{itemize}
        \item $B_c = \{ v \in V \mid v \text{ approves } c \text{ but not } p\}$
        \item $\Delta_c = \score{c} - \score{p}.$
    \end{itemize}
    For every $c\in C'$, we must delete votes that from $B_c$ whose total contribution to $c$’s score is at least $\Delta_c+1$, ensuring that, after the deletions, $p$ has strictly more approvals than any other candidate and thus becomes the unique approval winner. Let
    \[
      B = \{v \in V \mid v \text{ does not approve } p\}
          = \bigcup_{c \in C'} B_c .
    \]
    For every voter $v \in B$ introduce the binary variable
    \[
      x_v =
      \begin{cases}
        1 & \text{if voter } v \text{ is deleted},\\
        0 & \text{otherwise}.
      \end{cases}
    \]
    We now obtain the minimum number of deleted votes that required to make $p$ the unique
    approval winner by solving the following integer linear program:
    \[
    \begin{aligned}
    \text{Minimize} & \sum_{v\in B} x_v \\
    \text{subject to} \\
    &\sum_{v\in B_c} w(v)x_v \ge \Delta_c + 1
      && \text{for each } c\in C',\\
    & x_v \in \{0,1\}
      && \text{for each } v\in B .
    \end{aligned}
    \]

    Let us now argue the correctness of our algorithm.
    
    \begin{lemma}\label{lemma:cip-iff-approvalccdv}
        Let $I=(C,V,w,p)$ be an instance of \optapprovalccdvweighted\ and let $\calp = (A, b, c, d)$ be the covering integer program defined above. There exists a feasible integer solution to $\calp$ of objective value $k$ if and only if there exists a subset $V'\subseteq V$ with $\|V'\| = k$ such that $p$ is a unique approval winner of $(C, V-V')$. %
    \end{lemma}

    \begin{proof}
        For the left to right direction, let $\{x_v^*\}_{v\in B}$ be a feasible solution of $\calp$ with $\sum_{v\in B}x_v^*=k$ and choose $V'=\{v\in B\mid x_v^*=1\}$.
        For any $c\in C'$ we have
        \[
          \score{c}-\sum_{v\in B_c}w(v)x_v^*\le\score{c}-(\Delta_c+1)<\score{p},
        \]
        while the scores of candidates $C-C'$ are already strictly below that of $p$. Hence, after deleting votes $V'$, $p$ is a unique winner. The number of added votes equals the CIP objective.

        For the right to left direction, suppose $V'\subseteq V$ with $\|V'\|=k$ makes $p$ a unique winner when deleted. Necessarily $V'\subseteq B$, because deleting a voter who approves $p$ could not help $p$ win.
        For every $c\in C'$ we must have
        \[
          \!\!\sum_{v\in V'\cap B_c}w(v)\ge\Delta_c+1,
        \]
        Defining $x_v=1$ if $v\in V'$ and 0 otherwise yields a feasible integer solution of
        $\calp$ whose objective value is $k$.
        
        Thus $\calp$ has a solution of cost $k$ if and only if there exists a subset $V'\subseteq V$ with $\|V'\| = k$ such that $p$ is a unique approval winner of $(C, V-V')$.
    \end{proof}
    
    The resulting integer linear program is an instance of CIP which admits an $O(\ln m)$-approximation algorithm (and this factor is optimal unless $\P = \NP$)~\cite{kol-you:j:approx-cip}. By Lemma~\ref{lemma:cip-iff-approvalccdv}, solving this CIP will give the set of deleted votes for $\optapprovalccdvweighted$ thus showing the problem in \LOGAPX.

    Since the \optapprovalccdv\ is a special case of the \optapprovalccdvweighted\ by setting the output of the weight function $w$ to always be one, we have that \optapprovalccdv\ is also in \LOGAPX.
\end{proof}

\begin{theorem}
    Both \optapprovalccdv\ and \optapprovalccdvweighted\ are \LOGAPXCOM.
\end{theorem}
\begin{proof}
    Proposition~\ref{prop:opt-approval-ccdv-w-approx} establishes the $\LOGAPX$ membership.

    It is clear that \optapprovalccdv\ strict reduces to ${\rm opt}\dash{\rm approval}\dash{\rm CCDV}\dash{\rm weighted}$ and that both problems are in $\LOGAPX$, so it suffices to prove that MSC strict reduces to \optapprovalccdv. Let $\calt = \optapprovalccdv$.

        Let $f$ be a function that, given an instance of $(\calu,\cals)$ of MSC where\  $\calu = \{ x_1,..., x_n \}$, $\mathcal{S} = \{S_1,\dots, S_m\}$, and $S_i \subseteq \calu$ for each $i\in [m]$, outputs an instance $(C, V, p)$ of $\calt$:
    \begin{itemize}
        \item The candidate set is $C = \cals \cup \{p\}$, where $p$ is the distinguished candidate.
        \item The voter set is $V$ is defined as follows:
        \begin{enumerate}
            \item For each $i\in [m]$, there is one voter $w_i$ who approves every candidate in $S_i$ and disapproves $d$ as well as every candidate not in~$S_i$.
            \item Let \score{c} denote the current approval count of candidate $c$, and set $k = \max_{c\in C}\score{c}$. There are $k$ additional voters, each of whom approves \(p\) (and any subset of the other candidates as needed) so that, after their addition, every candidate has the same approval count \(k\).
        \end{enumerate}
    \end{itemize}
    For simplicity, we will say ``Type~$i$ votes'' to refer to those votes that are defined within item $i$ above.
    
    Let $g$ be a function that, given an instance $I = (\calu, \cals)$ of MSC and a solution $V'$ of $f(I) = (C, V, p)$ outputs $\widehat{\cals} = \{S_i \mid w_i \in V'\}$ aa a solution for instance $I$ w.r.t. MSC.

    Note that both $f$ and $g$ are polynomial-time computable. We now prove the correctness of our reduction.
    
    We prove that $\OPT_{\calt}(f(I)) = \OPT_{{\rm MSC}}(I)$. Let $\cals^*$ denotes an optimal solution to MSC on input $I$ and $V^*$ denotes an optimal solution to $\calt$ on input $f(I)$.

    First, we show that $\OPT_{\calt}(f(I)) \ge \OPT_{{\rm MSC}}(I)$. 
    Every candidate have the same approval count of $k$. We delete only votes that do not approve $p$, because deleting votes that approve $p$ would lower $p$’s approval count rather than that of its opponents, so we cannot delete any Type 2 votes. If $p$ can be made the unique approval winner, then for every candidate $c \in C - \{p\}$, we must delete at least one vote that approves $c$ but not $p$ from Type 1 votes, so that $p$’s approval count becomes strictly greater than that of $c$. Since we delete such a vote for each $c \in C - \{p\}$, the selected votes constitute a set cover of size $\card{V'}$ to $I$. 
    Thus $\OPT_{\calt}(f(I)) =  \card{V'} \ge \OPT_{{\rm MSC}}(I)$

    Next, we show that  $\OPT_{\calt}(f(I)) \le \OPT_{{\rm MSC}}(I)$. 
    Consider deleting the votes $V_{\cals^*} = \{w_i \mid S_i \in \cals^*\}$. Fix any candidate $x_j \in \calu$. Since $\cals^*$ covers $\calu$, there exists a $S_{i}\in \cals^*$ containing $x_j$, and the corresponding vote $w_{i}$ approves $x_j$ and disapproves $p$. Since all deleted votes disapprove $p$, $p$ is the unique approval winner. Thus we have produced a solution of size $\card{\cals^*}$ for $\calt$. Thus $\OPT_{{\rm MSC}}(I) = \card{\cals^*} \geq \OPT_{\calt}(f(I))$.

    Finally, let $V'$ be an arbitrary solution to $\calt$ on input $f(I)$. We now prove that $R_{{\rm MSC}}(I, g(I, V')) \leq R_\calt(f(I), V')$. Let $\cals' = g(I, V')$. By definition and the fact that $\OPT_{\calt}(f(I)) = \OPT_{{\rm MSC}}(I)$, it suffices to show that $\card{\cals'} \leq \card{V'}$. In the computation of $g$, we have $\card{\widehat{\cals}} = \card{V'}$, so $\card{\cals'} = \card{V'}$. Therefore $\card{\cals'} \leq \card{V'}$.
\end{proof}

Faliszewski, Hemaspaandra, and Hemaspaandra~\cite{fal-hem-hem:j:weighted-control} present approximation algorithms for constructive control by adding and deleting voters under $t$-approval and $t$-veto, with weighted votes by a reducing to Weighted Set Multicover, which is a special case of CIP\@.  Although 1-approval is plurality, their control problems (under plurality) are in $\P$, so their work does not subtract from ours.
Bredereck et al.~\cite{bre-fal-nie-sko-tal:j:mixinteger-lp} essentially gave the analogous result to Propositions~\ref{prop:opt-approval-ccav-w-approx} and~\ref{prop:opt-approval-ccdv-w-approx} in the \emph{nonunique-winner} model (recall that this paper is in the unique-winner model); in their work, they gave polynomial-time many-one reductions from those problems to the Uniform Multiset Multicover (another special case of CIP), and we mention that their reductions are in fact strict reductions. 
Unaware of their construction, we gave our result in terms of CIP, but we mention that our approach and that of Bredereck et al.~\cite{bre-fal-nie-sko-tal:j:mixinteger-lp} are similar in flavor.

\subsection{Plurality and Candidate Deletion}

In this section, we prove that constructive control by deleting candidates under plurality (both with weighted and unweighted votes) is approximable by giving a simple and greedy $O(m)$-approximation algorithm, where $m = \card{C}$. Moreover, we give a fine-grained conditional lower bound of $\Omega(m^{1/4})$. %

To give the upper bound (i.e., the approximability result) we  appeal to the notion of voicedness: a voting rule is said to be voiced if in any one-candidate election, that voting rule elects that candidate as the unique winner. Every voting rule in this paper is clearly voiced.

\begin{lemma}\label{lemma:ccdc-approx}
    Let $\cale$ be a voiced voting rule.
    Then
     constructive control by deleting candidates under $\cale$, both with weighted and unweighted votes,
    is in $\POLYAPX$, and each problem admits an $O(m)$-approximation.
\end{lemma}
\begin{proof}
    Let $\calt$ denote the optimization problem for constructive control by deleting candidates under $\cale$ with weighted votes,
    let $I = (C, V, w, p)$ be an instance of $\calt$, and let $m=\|C\|$\@. Our $O(m)$-approximation algorithm for $\calt$ simply outputs $C-\{p\}$.

    Since $\cale$ is voiced,
    it's clear that the algorithm always outputs a correct solution and runs in polynomial time. Let us now prove that the algorithm is an $O(m)$-approximation algorithm.
    By definition, $\OPT_{\calt}(I) \geq 1$, so the performance ratio is 
        $$\frac{\mu_{\calt}(I, C-\{p\})}{\OPT_{\calt}(I)} = \frac{m}{\OPT_{\calt}(I)} \leq m,$$
    which concludes our proof that the algorithm is an $O(m)$-approximation algorithm. The $\POLYAPX$ membership follows from the fact that $m\leq |I|$.

    This also yields an $O(m)$-approximation algorithm for the unweighted version of $\calt$.
\end{proof}

The scope of our paper is to determine which NP-hard electoral control problems (both with weighted and unweighted votes) have approximation algorithms; for constructive control by deleting candidates under plurality, we have achieved our goal. However to supplement that result,  we provide a lower bound by leveraging the best-known (conditional) lower bound on the Minimum $k$-Union problem (defined below).

\OPTDEF{Minimum $k$-Union (M$k$U)}
{A nonempty set $\calu = \{u_1, \ldots, u_n\}$, a family $\cals = \{S_1, \ldots, S_m\}$ of nonempty subsets of $\calu$, and a natural number $k \leq m$.}
{A set $\cals' \subseteq \cals$ of size $k$.}%
{The value $\|\bigcup_{S \in \cals'} S\|$.}

It is known that assuming the so-called ``Hypergraph Dense versus Random'' Conjecture (HDR), the lower bound on the approximation ratio for $\MkU$ is $\Omega(m^{1/4})$~\cite{chl-din-mak:c:min-k-union}.\footnote{The definition of HDR is beyond the scope of this paper, so we refer interested readers to the work of Chlamt{\'a}{\v{c}}, Dinitz, and Makarychev~\cite[Conjecture~1.5]{chl-din-mak:c:min-k-union} for details on it.} 
It is also known that MkU is at least as hard to approximate as MSC, so the approximation lower bound on the approximation ratio for $\MkU$ is in $\Omega(\log{n})$, unconditionally.

We next give a strict reduction from M$k$U to opt-$\pluralityccdc$, thereby showing that giving a better-than-$O(m^{1/4})$-approximation algorithm for $\pluralityccdc$ would resolve HDR\@.

\begin{theorem}\label{t:mku-ccdc}
    $\MkU \strictred \optd\pluralityccdc$.
\end{theorem}
\begin{proof}
    For brevity, let $\calt = \optd\pluralityccdc$\@.

    Let $f$ be a function that does the following
    given an $\MkU$ instance\, i.e., $\calu = \{u_1, \ldots, u_n\}$, a family $\cals = \{S_1, \ldots, S_m\}$ of nonempty subsets of $\calu$, and a positive integer $k \leq m$. 
    We assume without loss of generality that $k > 1$ as otherwise we can solve the problem in polynomial time.
    Let $B = \{b_1, \ldots, b_{n+1}\}$ be a set of buffer candidates, let $N_i$ be the number of sets in $\cals$ containing $u_i$, and let $N = \max\{N_1, \ldots, N_n\}$. Our reduction then outputs $(C, V, p)$, where
    \begin{itemize}
        \item  $C = \calu \cup B \cup \{p\}$ is the set of candidates;
        \item $V$ is the collection of votes comprising of the following 
        votes:
        \begin{enumerate}
            \item for each $i \in [m]$, one vote of the form $ S_i > p > \cdots$,
            \item $k-1 + N$ votes of the form $b_1 > \cdots > b_{n+1}  > \cdots > p$, and
            \item $N$ votes of the form $p> \cdots$; and %
        \end{enumerate}
        \item $p$ is the distinguished candidate.
    \end{itemize}

    For simplicity, we will say ``Type~$i$ votes'' to refer those votes that are defined within item $i$ above.

    Let $g$ be a function that given an instance $I = (\calu, \cals, k)$ of $\MkU$ and a solution $D$ of $f(I) = (C, V, p)$ outputs a solution for $I$ w.r.t.\ $\MkU$ as follows. Let $\widehat{\cals} = \{S_i \mid S_i \subseteq D\}$. Then by Lemma~\ref{lemma:enough-candidates-deleted}, $\widehat{\cals}$ contains a $k$-sized subset and any such subset suffices as an output of $g$.

    \begin{lemma}\label{lemma:enough-candidates-deleted}
        $\|\widehat{\cals}\| \geq k$.
    \end{lemma}
    \begin{proof}
        Suppose not, for the sake of contradiction. Then $p$ is not ranked first by any of the Type~2 votes, but some candidate $d \neq p$ is. Thus for $p$ to defeat $d$, $p$ must be ranked first by at least $k$ Type~1 votes. Those votes correspond exactly to $\widehat{\cals}$, and so we have a contradiction. 
        \renewcommand{\qedsymbol}{(Lemma~\ref{lemma:enough-candidates-deleted})~\ensuremath{\openbox}}
    \end{proof}

    Note that $f$ and $g$ are polynomial-time computable. We'll now prove the correctness of our reduction.
    Let $I = (\calu, \cals, k)$ be an arbitrary instance of $\MkU$ and
    let $f(I) = (C, V, p)$.
    We will  prove that {$\OPT_{\calt}(f(I)) = \OPT_{\MkU}(I)$}\@.
    Let $\cals'$ denote an optimal solution to $\MkU$ on input $I$ and let $D^*$ denote an optimal solution to $\calt$ on input $f(I)$. Furthermore, let $D' = \bigcup_{S \in \cals'} S$. 
     
     We will first show that $\OPT_{\MkU}(I) \leq \OPT_{\calt}(f(I))$. 
     Assume for the sake of contradiction that $\|D'\| > \|D^*\|$. Because $p$ is a unique winner of $(C-D^*, V)$ and $D$ is an optimal solution, it must be that $D^* \cap B = \emptyset$. Indeed,
     deleting all the candidates in $\calu$ always leads to a solution and since $\|\calu\| \leq n$, it follows that $\|D^*\| \leq n$.
     Moreover, the only way deleting candidates from $B$ to help $p$ win is if all candidates in $B$ are deleted, which would contradict the fact that $\|D^*\| \leq n$. 
     So the only way for $p$ to defeat $b_1$ is if at least $k + N$ votes  rank $p$ first, and for $p$ to be guaranteed to defeat each $u \in \calu$ at least $N+1$ votes must rank $p$ first. 
     Since $k \geq 1$, it suffices to consider the fact that $p$ defeats $b_1$. 
     In this case, at least $k$ Type~1 votes had all the candidates ranked above $p$ deleted in election $(C-D^*, V)$, i.e., $\widehat{\cals} = \{S_i \mid S_i \subseteq D^*\}$ and $\|\widehat{\cals}\| \geq k$. Taking the union of any $k$-sized subset of $\widehat{\cals}$ yields a set of size at most $\|D^*\| < \|D'\|$ and thus better solution to $\MkU$ on input $I$, which is a contradiction.

    We now show that $\OPT_{\calt}(f(I)) \leq \OPT_{\MkU}(I)$.
    Suppose for the sake of contradiction that  $\|D'\| < \|D^*\|$ and
     consider the plurality election $E = (C-D', V)$. In that election, $\myscore_E(p) = k+N$; $\myscore_E(b_1) = k-1+N$;  for each $u \in \calu$, $\myscore_E(u) \leq N$; and for each remaining candidate $d$, $\myscore_E(d) = 0$.
    Thus $p$ is a unique winner of $E$, and so $D'$ is a solution for $f(I)$ w.r.t.\ $\calt$. But by assumption $\|D'\| < \|D^*\|$, which leads to a contradiction.

    For the rest of this proof, let $D$ be an arbitrary solution to $\calt$ on input $f(I)$ (recall that a simple corollary to Lemma~\ref{lemma:ccdc-approx}'s proof is that such a set is guaranteed to exist).

    We now prove that $R_{\MkU}(I, g(I, D)) \leq R_\calt(f(I), D)$. Let $D' = \bigcup_{S \in g(I, D)}S$.
    By definition and the fact that $\OPT_{\calt}(f(I)) = \OPT_{\MkU}(I)$, it suffices to show that $\|D'\| \leq \|D\|$.
    Notice that in the computation of $g$, every element of $\widehat{\cals}$ is a subset of $D$, so every element of $D'$ must also be an element of $D$. Therefore $\|D'\| \leq \|D\|$.
\end{proof}

Because our reduction is strict, it is easy to see than an $f(m)$-approximation algorithm for $\optd\pluralityccdc$ implies an $f(m)$-approximation algorithm for $\MkU$, so we are able to state the following result.

\begin{theorem}
Assuming the Dense versus Random Conjecture, the approximation ratio of any approximation algorithm for constructive control by deleting candidates under plurality is in $\Omega(m^{1/4})$, and this holds for both weighted and unweighted votes.
\end{theorem}
\begin{proof}
    This follows directly from the conditional lower bound given on $\MkU$ by Chlamt{\'a}{\v{c}}, Dinitz, and Makarychev~\cite{chl-din-mak:c:min-k-union} and from our Theorem~\ref{t:mku-ccdc}. This lower bound also holds when the votes are weighted.
\end{proof}

While we would have preferred to have give an unconditional lower bound,
providing a lower bound based on HDR is reasonable and has been used in other work on approximations (e.g., see~\cite{che-che-ye-zha:c:min-k-union-extensions}).
\section{Inapproximability}\label{s:inapprox}

In this section, we present our inapproximability results. %
  All of the results in this section follow the same strategy, namely we show for each inapproximable problem that even finding \emph{one} solution for any input is $\NP$-hard.

To do so for nonpartition-based problems, we consider their ``unlimited'' variants. In particular, each of the nonpartition-based problems can be viewed as \emph{limited} forms of addition/deletion since there is a nonnegative integer $k$ that limits the number of candidates/votes that can be added/deleted. In the \emph{unlimited} variants of these problems, there is no such restriction---any number of candidates/votes that can be added/deleted, e.g., all unregistered votes or all candidates may be affected. Such variants are denoted by adding a ``U'' before the control action, e.g., plurality-CCUAC is the unlimited version of plurality-CCAC\@.\footnote{This problem, i.e., \pluralityccuac, is precisely the one used by~\cite{bar-tov-tri:j:control} and by many early papers in the field when modeling constructive control by adding candidates. However, their definition was asymmetric with the other addition/deletion control problems they defined, and the modern, limited version of constructive control by adding candidates, which is the version considered in this paper, was introduced by Faliszewski et al.~\cite{fal-hem-hem-rot:j:llull}.}

\subsection{General Cases}
Our first theorem in this section, albeit simple in its statement and its proof, is a useful one as it establishes in one fell swoop the inapproximability of a large number of control problems. Since we are concerned with closing all the gaps opened by Brelsford~\cite{bre:msthesis:approximations}, we have opted to still keep this result, despite its simplicity.

\begin{theorem}\label{t:inapproximable-partitioning}
    For 
    each partition-based decision control problem $\calt$ (both with weighted and unweighted votes),
    if $\calt$ is \NP-complete, then the optimization version of $\calt$, regardless of the measure used, is not approximable, unless $\P = \NP$.
\end{theorem}
\begin{proof}
    Notice that for any partition-based control problem $\calt$, an approximate solution is still a valid solution under any valid measure as there are no restrictions on the solution. Thus an approximate algorithm for $\calt$ is also an exact polynomial-time algorithm for $\calt$\@. If $\calt$ is $\NP$-complete, this implies that $\P=\NP$. Therefore, unless $\P=\NP$, no such approximation algorithm exists.\footnote{As an aside, the result that Theorem~\ref{t:inapproximable-partitioning} provides also follows from Theorem~5 of \cite{car-cha-hem-nar-tal-wel:j:search-versus-search}; using the language of that paper, an approximation algorithm for an $\NP$-complete partition-based control problem would indeed be a polynomial-time computable ``refinement'' of the corresponding search problem, and thus by their Proposition~3 it would follow that $\SAT \in \P$\@. Thus unless $\P=\NP$, any $\NP$-complete partition-based control problem is inapproximable. However, we favor the proof we give as it highlights a simplicity in the argument, and that same simple argument line proves to be useful in the rest of the section.}
\end{proof}

A natural follow-up to the above result is whether the partition-based problems we define are still inapproximable if one were to consider their ``limited'' version where the problem's input includes an integer bound on the absolute difference between the sizes of the parts of the partition. However, it's easy to see that the optimization problems for the limited cases coincide with the optimization versions of the corresponding standard partitioning problems, and thus the problems remain inapproximable.

\subsection{Plurality}
To prove that destructive control by deleting candidates under plurality is inapproximable, we modify the proof that \pluralitydcdc\ is \NP-complete by \cite{hem-hem-rot:j:destructive-control} to hold in the ``unlimited'' setting, i.e., we prove that destructive control by unlimited deleting of candidates 
under plurality is \NP-complete.

\begin{proposition} 
\label{p:pluralitydcudc-npcomplete}
$\pluralitydcudc$ (destructive control by unlimited deleting of candidates under plurality) is \NP-complete.
\end{proposition}

\begin{proof}

To prove our result, we give a reduction from Hitting Set, which is a standard NP-complete problem~\cite{kar:b:reducibilities}, and a construction given by Hemaspaandra, Hemaspaandra, and Rothe~\cite{hem-hem-rot:j:destructive-control}.

\LANGDEF{Hitting Set}
{A finite set $B = \{b_1, b_2, \dots, b_m\}$, a family $\mathcal{S} = \{S_1, S_2, \dots, S_n\}$ of subsets $S_i \subseteq B$, and a positive integer $k$.}
{Does there exist a subset $B' \subseteq B$ with $\|B'\| \leq k$ such that for every $i \in [n]$, $S_i \cap B' \neq \emptyset$?}

\begin{construction}
[\cite{hem-hem-rot:j:destructive-control}] 
Given a triple $(B,S,k)$, where $B = \{b_1, b_2, \dots, b_m\}$ is a set, $S = \{S_1, S_2, \dots, S_n\}$ is a family of subsets $S_i$ of $B$, and $k \le m$ is a positive integer, we construct the following election:
\begin{itemize}
    \item The candidate set is $C = B \cup \{c, w\}$.
    \item The voter set $V$ is defined as follows:
    \begin{enumerate}
        \item There are $2(m - k) + 2n(k + 1) + 4$ votes of the form $c > w > \dots$.
        \item There are $2n(k +1) + 5$ votes of the form $w > c > \dots$.
        \item For each $i$, $1 \le i \le n$, there are $2(k+1)$ votes of the form $S_i > c > \dots$, where ``$S_i$'' denotes the elements of $S_i$ in some arbitrary order.
        \item Finally, for each $j$, $1 \le j \le m$, there are two votes of the form $b_j > w > \dots$.
    \end{enumerate}
    \item The distinguished candidate is c.
\label{hhr07construction}
\end{itemize}
\end{construction}

    Given a subset of candidates to delete $C' \subseteq C - \{c\}$, we can compute the plurality score in polynomial time, and then check whether $c$ is a winner in polynomial time. Thus, the problem is in NP\@. Using Construction \ref{hhr07construction}, we will now show the reduction from Hitting Set to \pluralitydcudc.\ In this proof, Group~($i$) refers to collection of votes defined by item~$i$ in the construction above.

    First, assume there exists a hitting set $B' \subseteq B$ is a hitting set with $\|B'\| \le k$. Let $D = B - B'$ is the set of candidates to be deleted. Then, consider the election $(C - D, V)$:
    \begin{itemize}
        \item From Group $(1)$, there will be $2(m-k) + 2n(k+1) + 4$ votes for $c$.
        \item From Group $(2)$, there will be $2n(k+1) + 5$ votes for $w$.
        \item For each $1 \le i \le n$, there exists $b_j \in S_i$ since $D = B - B'$. Therefore, from Group $(3)$ there is no votes for $c$ (these votes will goes to some $b_j$ instead).
        \item Since we deleted $(m-k)$ candidates $b_j$, $w$ received an additional $2(m-k)$ votes from Group $(4)$. The remaining votes goes to some $b_j$.
    \end{itemize}
    Thus, $\score{w} = 2(m-k) + 2n(k+1) + 5  > \score{c} =  2(m-k) + 2n(k+1) + 4$ votes. Therefore, $w$ is a unique winner of election $(C-D, V)$, and so there is a successful control action.
    
    Second, assume there exists a deletion set $D \subseteq C - \{c\}$ such that $w$ is the unique plurality winner of $(C - D, V)$. Let $B' = B - D$ be candidates that are not deleted.
    
    If there exists a $S_i$ such that $S_i \cap B = \emptyset$, we know that $c$ got additional $2(k+1)$ votes from Group (3). Let $\alpha$ be the number of those ``nonhitting'' sets, then $c$ will have $ 2(m-k) + 2n(k+1) + 4 + 2\alpha(k+1)$ votes. However, even when $w$ received all votes from Group~(2) and~(4), can only have at most $2n(k+1) + 5 + 2m$ votes.  Therefore, $w$ can be the unique winner only when $\alpha = 0$, meaning that every $S_i$ intersects $B'$, i.e., $B'$ is a hitting set of $S$.
    
    If $\|B'\| > k$, then $\|D \cap B\| < (m-k)$, so 
        $$
        \score{w} < 2n(k+1) + 5 + 2(m-k)
        $$
        so $c$ is still a winner. Therefore, we know that $\|B'\| \le k$.
    Thus, $B'$ is a hitting set of $S$ of size less than or equal to $k$.
\end{proof}

Using a similar argument as the one used in the proof of Theorem~\ref{t:inapproximable-partitioning}, we prove that destructive control by deleting candidates under plurality is not approximable, unless $\P=\NP$.

\begin{theorem}
    Both $\optpluralitydcdc$ and $\optpluralitydcdcweighted$ are inapproximable, unless $\P = \NP$.
\end{theorem}
\begin{proof}
    By Proposition~\ref{p:pluralitydcudc-npcomplete}, $\pluralitydcudc$ is $\NP$-complete. Any approximation algorithm for $\optd{\rm plurality}\dash\allowbreak{\rm DCDC}$ implies that $\pluralitydcudc \in \P$, thus unless $\P=\NP$, no such approximation exists.
    For the weighted version, we can construct a reduction from the unweighted version by assigning unit weights to all votes (i.e., each weight is set to 1). As a result, the corresponding weighted decision problem
    is \NP-complete and thus its corresponding approximation version is also inapproximable.
\end{proof}

\subsection{Condorcet}

In this section we prove that the optimization problems of constructive control by adding/deleting voters, both in the weighted and unweighted setting, under Condorcet voting are inapproximable. 

    Bartholdi, Tovey, and Trick~\cite{bar-tov-tri:j:control} showed that \condorcetccav\ is \NP-complete and we modify their reduction to hold in the unbounded setting, thereby proving that \condorcetccuav\ remains \NP-complete when an unlimited number of voters may be added. 

\begin{proposition}\label{p:condorcet-ccuav}
        \condorcetccuav\ (constructive control by unlimited adding of voters under Condorcet) is \NP\dash complete.
\end{proposition}
\begin{proof}
    \condorcetccuav\ is in \NP\ because given a set of added votes, we can determine in polynomial time if a candidate is the Condorcet winner. To establish \NP\dash hardness, we give a reduction from X3C to \condorcetccuav. 
    
    \LANGDEF{Exact Cover By Three-Sets (X3C)}
{A finite set $B = \{b_1, b_2, \dots, b_{3k}\}$, where $k$ is a positive integer, and  a family $\mathcal{S} = \{S_1, S_2, \dots, S_n\}$ of subsets $S_i \subseteq B$ with $\|S_i = 3\|$ for each $i$.}
{Does there exist a subset $\cals' \subseteq \cals$ with $\|\cals'\| =k$ and $\bigcup_{ S_i \in \cals'} = B$?}  
    
    Given an instance $(B, \cals)$ of X3C, where $B = \{b_1, b_2, \dots, b_{3k}\}$ and $\mathcal{S} = \{S_1, S_2, \dots, S_n\}$ with $S_i \subseteq B$ and $\|S_i\| = 3$ for every $i$. 
    For each $S_i$, we will denote its members by $b_i^1$, $b_i^2$, and $b_i^3$.
    Without loss of generality, we assume that $k \ge 3$ (if $k < 3$ then we can solve the problem in polynomial time). We construct the following instance of \condorcetccuav:
    \begin{itemize}
        \item The candidate set $C = B \cup \{c, p\}$, where $p$ is the distinguished candidate.
        \item The collection of registered votes $V$ consists of:
        \begin{itemize}
            \item $k-1$ votes of the form $B>p>c$.
            \item 2 votesof the form $p>c>B$.
        \end{itemize}
        \item The collection of unregistered votes $W$ consists of: For each $S_i$, one vote of the form $b_i^1>b_i^2>b_i^3>c>p>B-S_i$. %
    \end{itemize}

    For candidates $a$ and $b$, let $A=\{v\in V \mid v \text{ prefers } a \text{ to } b\}$ be the set of votes that prefer $a$ to $b$, and $B=\{v\in V \mid v \text{ prefers } b \text{ to } a\}$ be the set of votes that prefer $b$ to $a$. We define the pairwise score as $\score{a,b}= \|A\| - \|B\|$.
    We claim that $\cals$ contains an exact cover of $B$ by 3-sets if and only if $p$ can be made the Condorcet winner of $(C, V)$ by adding any number of voters from $W$.
    
    For the left to right direction, add the $k$ votes corresponding to the exact cover. All such votes rank
    candidate $b_i$ above $p$ exactly once, while the remaining $k-1$ added votes rank $p$ above that same $b_i$. Thus 
    \[\score{p, b_i} = -(k-3) + (k-1) - 1 = 1.\]

    Hence, p beats every candidate $b_i$ by one vote. Because all $k$ added votes prefer $c$ above $p$,  we have $\score{p, c} = (k+1) - k = 1$ and $p$ also beats $c$ by one vote. Therefore, $p$ becomes the Condorcet winner after adding $k$ votes.

    For the right to left direction, assume there is a set $W' \subseteq W$ whose addition makes $p$ the Condorcet winner. Since every added vote ranks $c$ above $p$, adding more than $k$ votes will result in $c$ beating $p$ in a pairwise contest, so $\card{W'} \leq k$. %
    No candidate $b_i$ can be preferred over $p$ by more than one added vote. Otherwise, $b_i$ would gain at least two additional points against $p$ (so $\score{b_i,p}\ge (k-3) + 2 - (k-1) \ge 1$), while $p$ would gain at most $k-2$ points, and $p$ would lose the pairwise contest with~$b_i$. Thus each $b_i$ is ranked above $p$ by at most one vote in~$W'$. Assume, for contradiction, that some $b_i$ is ranked above $p$ by no added vote. There are $3k$ candidates $B = \{b_1,\dots,b_{3k}\}$ and $k$ votes in~$W'$, each contributing exactly three positions above $p$. By the pigeon-hole principle, some other candidate $b_{i'}$ must be
    ranked above $p$ by at least two votes, contradicting the bound of one vote per $b_i$. Therefore every $b_i$ is ranked above $p$ by exactly one of the $k$ added votes.
    The $k$ added votes in $W'$ correspond to $k$ disjoint sets $S_i$ whose union are $B$. Hence the set of added votes is an exact cover of $B$ by 3-sets.
\end{proof}

\begin{theorem}\label{t:condorcet-ccav-inapp}
    \optcondorcetccav\ and \optd\allowbreak\condorcet\dash\allowbreak\CC\AV\dash\weighted\ are inapproximable, unless $\P = \NP$.
\end{theorem}
\begin{proof}
    By Proposition~\ref{p:condorcet-ccuav}, $\condorcetccuav$ is $\NP$-complete. 
    An approximation algorithm for 
    $\optd\condorcet\dash\allowbreak{\rm CCAV}$
    implies that $\condorcetccuav\in \P$, thus unless $\P=\NP$, no such approximation exists.
    For $\optd\allowbreak\condorcet\dash\allowbreak\CC\AV\dash\weighted$, we can construct a reduction from $\optd\allowbreak\condorcet\dash\allowbreak\CC\AV$ by assigning unit weights to all votes (i.e., each weight is set to one). As a result, \optcondorcetccavweighted\ is also $\NP$-hard and thus inapproximable.
\end{proof}

Next, we show that \condorcetccudv\ is \NP-complete. In the case of \condorcetccuav, we were able to modify the known reduction from X3C to \condorcetccav\ to establish \NP-hardness. However, for \condorcetccudv, we were not able to modify the known reduction from X3C to \condorcetccdv\ (see~\cite{fal-hem-hem-rot:j:llull}), as it does not allow us to enforce a constraint on the number of deletable votes. We instead give a new reduction from Hitting Set, which enables us to satisfy all required constraints for our proof.

This next result is surprising as one typically expects the ``unlimited'' deletion of voter to be polynomial-time computable, with the ``obvious'' strategy being to delete all voters. Such a strategy certainly works under many voting rules like approval and plurality, but under Condorcet doing so does not make the distinguished candidate a Condorcet winner, and in fact, simply finding a successful deletion of voters is hard.

\begin{proposition}\label{p:condorcet-ccudv}
    \condorcetccudv\ (constructive control by unlimited deleting of voters under Condorcet) is \NP\dash complete.
\end{proposition}
\begin{proof}
    \condorcetccudv\ is in \NP\ because given a set of deleted votes, we can determine in polynomial time if a candidate is a winner. To establish \NP\dash hardness, we give a reduction from Hitting Set to \condorcetccudv.
    Given an instance of $(B, \cals, k)$ Hitting Set, where $B = \{b_1, b_2,\dots,b_m\}$, a family $\cals = \{S_1, S_2,\dots,S_n\}$ of subsets $S_i \subseteq B$, and a positive integer $k$. Without loss of generality, we assume $k \leq m$ (if $k > m$ then we can verify if each element $b_i$ lies in some set $S_j$ and either choose those $m$ sets or answer No). We construct the following instance of \condorcetccudv :
    \begin{itemize}
        \item The candidate set $C$ consists of the $n+k+3$ candidates:
        \begin{itemize}
            \item The distinguished candidate $p$.
            \item There are $n$ candidates $\cals = \{S_1, S_2,\dots,S_n\}.$
            \item There are $k+1$ candidates $D = \{d_1, d_2,\dots,d_{k+1}\}$.
            \item There is one candidate $e$.
        \end{itemize}
        \item  The voter set $V$ consists of $m+k+1$ votes:
        \begin{itemize}
            \item There is one voter $x_1$ that has a preference 
                $\cals > D - d_1 > p > d_1 > e$.
            \item For each $2 \leq i \leq k + 1$, there is one voter $x_i$ that has a preference:
                $D-d_i>p>d_i>e>\cals$.
            \item For each element $b_j(1 \leq j \leq m)$, let $\cals'_j = \{S_\ell \mid b_j \in S_\ell\}$, there is one voter $y_j$ that has a preference:
            $\cals - \cals'_j>e>p>D>\cals'_j$.
        \end{itemize}
    \end{itemize}
    
    We claim that there exists a subset $B' \subseteq B$ that is a hitting set of $B$ with $\|B'\| \leq k$ if and only if there exists a set of votes $V'$ such that $p$ is the Condorcet winner of the election $(C, V - V')$. Let $X = \{x_i \mid 1\leq i\leq k+1\}$ and $Y = \{y_j \mid 1\leq j\leq m\}$. For candidates $a,b$, 
    let $N_V(a, b)$ be the number of votes in $V$ where $a$ is ranked above $b$.
    We define the pairwise score of the head-to-head contest between $a$ and $b$ as $\score{a,b}= N_V(a, b) - N_V(b, a)$.

    For the left to right direction, assume $B' \subseteq B$ is a hitting set of $B$ and $\|B'\| \leq k$. Keep all votes in $X$ and every voter $y_j$ such that $b_j \in B'$ and delete the remaining $m-k$ votes in $Y$. Since $B'$ is the hitting set, there will be at least one voter in $Y$ that preferred $p$ over $S_i$ so 
    \[\score{p, S_i} \geq (-1) + k + 1 - (k-1) \geq 1.\]
    
    Among all votes in $X$, each candidate $d_i$ is lower higher than $p$ only once, so 
    \[\score{p, d_i} = 1 + k - k = 1 \text{ and}\]
    \[\score{p, e} = k + 1 - k = 1.\]

    Thus $p$ beats every candidate and is the Condorcet winner after deleting $m - k$ votes.

    For the right to left direction, let $V'$ be the set of deleted votes. For an arbitrary candidate $d_i$, if the corresponding voter $x_i$ is deleted, then in any remaining votes would rank either $d_i$ or $e$ above $p$, with the other below $p$. Consequently, $\score{p, d_i} = \score{e, p}$ and $p$ could not beat both $d_i$ and $e$. Repeating this argument for every $d_i$, we conclude that no voter in $X$ can be deleted. Counting only the $k+1$ votes in $X$, we have $\score{p,e} = k + 1$. Every voter in $Y$ prefers $e$ to $p$ so at most $k$ votes in $Y$ can be kept.

    Similarly, if we count only the $k+1$ votes in $X$, we have $\score{d_i, p} = k - 1$ so at least $k$ votes in $Y$ must be kept since every voter in $Y$ prefers $p$ over $d_i$. Therefore, we keep exactly $k$ votes in $Y$ and \[\|V'\| = m + k + 1 - (k+1) - k = m - k.\]

    Considering only the votes in $X$, we have $\score{p, S_i} = k - 1$ for every $S_i$. If, for some $S_i$, all $k$ retained votes in $Y$ prefer $S_i$ over $p$, then $\score{p, S_i} = -1$ and $S_i$ would beat $p$ in a pairwise contest. Therefore, the set of votes $Y' \subseteq Y$ that is kept must contain, for each $S_i$, at least one voter $y_j$ with $b_j\in S_i$. The corresponding elements of~$B$ therefore form a hitting set of size $k$.
\end{proof}

\begin{theorem}
    $\optcondorcetccdv$ and $\optd\condorcet\dash\allowbreak \CC\DV\dash\allowbreak \weighted$ are inapproximable, unless $\P = \NP$.
\end{theorem}
\begin{proof}
    By Proposition~\ref{p:condorcet-ccudv} $\condorcetccudv$ is $\NP$-complete. 
    Any approximation algorithm for $\optd\condorcet\dash\allowbreak{\rm CCDV}$ implies that $\condorcetccudv\in \P$, thus unless $\P=\NP$, no such approximation exists.
    For $\optd\condorcet\dash\allowbreak \CC\DV\dash\allowbreak \weighted$, we can construct a reduction from $\optd\condorcet\dash\allowbreak \CC\DV$ by assigning unit weights to all votes (i.e., each weight is set to 1). As a result, $\optd\condorcet\dash\allowbreak \CC\DV\dash\allowbreak \weighted$ is also $\NP$-hard and thus inapproximable.
\end{proof}

\section{Related Work}\label{s:related-work}

The study of electoral control problems and their complexity was initiated by Bartholdi, Tovey, and Trick~\cite{bar-tov-tri:j:control}, who in their seminal work defined the ``constructive'' setting. They defined the notions of vulnerability, immunity, and resistance in that context, and provided certain axiomatic results.
Hemaspaandra, Hemaspaandra, and Rothe~\cite{hem-hem-rot:j:destructive-control} continued this line of research by defining the ``destructive'' setting, clarifying the winner model, and defining the standard tie-handling rules for the partition-based control problems. They also continued the vulnerable/immune/resistant analysis in their work. Since then, a lot of work has gone into defining new types of control and proving complexity-related results for those problems.
For a more comprehensive overview of the literature on electoral control, see the recent chapter by Baumeister and Rothe~\cite{bau-rot:b:preference-aggregation}.

As mentioned earlier, much work on approximation algorithms (for optimization problems) in the study of electoral manipulative attacks has focus on manipulation and bribery.
Brelsford et al.~\cite{bre-fal-hem-sch-sch:c:approximability-of-manipulation} studied the approximability of manipulating elections under the Borda rule, establishing APX-hardness for some manipulation problems. For bribery, various approximation results have been established across different voting systems. Faliszewski, Hemaspaandra, and Hemaspaandra~\cite{fal-hem-hem:j:bribery}, Keller, Hassidim, and Hazon~\cite{kel-has-haz:j:approximating,kel-has-haz:j:approximate-weighted-priced-bribery}, and Xia~\cite{xia:margin-of-victory} analyze the complexity and approximability of bribery in scoring rules and margin of victory computations. Yang~\cite{yan:c:election-few-candidates} and Faliszewski et al.~\cite{fal-hem-hem-rot:j:llull} show that bribery is fixed-parameter tractable in settings with few candidates. For a comprehensive summary of parameterized approaches to manipulation, control, and bribery, see the survey by Betzler et al.~\cite{bet-bre-che-nie:c:parameterized-elections-survey}.

Some early explorations into approximation algorithms for electoral control were by Brelsford~\cite{bre:msthesis:approximations}, who studied constructive and destructive control by adding candidates and bribery under plurality voting, proving inapproximability for both the weighted and unweighted cases. 
Brelsford showed that under plurality, both constructive and destructive control by adding candidates is inapproximable, both with weighted and unweighted votes.
Faliszewski, Hemaspaandra, and Hemaspaandra~\cite{fal-hem-hem:j:weighted-control} present an approximation algorithm for constructive control by adding and deleting votes under $t$-approval and $t$-veto, with weighted votes. They do so by a reduction to Weighted Set Multicover, which is also a special case of CIP, and also leverage the results from~\cite{kol-you:j:approx-cip}. Even though plurality is 1-approval, it is known that plurality is vulnerable to both control by adding and deleting votes (see~\cite{bar-tov-tri:j:control, hem-hem-rot:j:destructive-control}). Thus their work does not subtract from our results.
Bredereck et al.~\cite{bre-fal-nie-sko-tal:j:mixinteger-lp} on the other hand provide interesting results 
variations of multiset cover (which are all special cases of CIP), and prove analogous results to our Propositions~\ref{prop:opt-approval-ccav-w-approx} and~\ref{prop:opt-approval-ccdv-w-approx}, but in the nonunique-winner model; however, their results are not explicitly about approximations and we note that their polynomial-time many-one reductions are effectively strict reductions.
Some other work in the last few years has also focused on the search complexity of control problems~\cite{hem-hem-men:j:search-versus-decision,car-cha-hem-nar-tal-wel:j:search-versus-search}.

Beyond approximation with respect to optimization problems, a large body of work
investigates fixed-parameter tractability (FPT). Many FPT results exist for control and bribery under various voting rules in the weighted setting. The techniques used often reduce election problems to covering or packing problems and employ integer programming formulations with convex/concave constraints (see for example~\cite{bre-fal-nie-sko-tal:j:mixinteger-lp}). Multiprong control---simultaneous application of multiple control types---has also been shown to admit FPT algorithms in specific settings~\cite{fal-hem-hem:j:multiprong}.

Our work gives a strict reduction from constructive control by deleting candidates under plurality voting to the Minimum \(k\)-Union ($\MkU$) problem. The best known approximation hardness for $\MkU$ is conditioned on the ``Dense vs. Random Conjecture,'' which is formalized in the work of Chlamt{\'a}{\v{c}}, Dinitz, and Makarychev~\cite{chl-din-mak:c:min-k-union}. Under this conjecture, it is believed that $\MkU$ cannot be approximated within a factor better than $O(m^{1/4})$, and that lower bound still stands to date. On a related note, 
the Maximum $k$-Cover, which is the maximization version of $\MkU$, is known to not have a better-than-$O(\log n)$ approximation~\cite{fei:j:set-cover}. Minimum $k$-Union and Maximum $k$-Cover are closely related to each other, and while the latter is well studied, there is a small amount of literature on the former. However, a natural consequence of their connection is that unless $\P=\NP$, any approximation algorithm for $\MkU$ has an $\Omega(\log n)$ lower bound.

Our work continues this line of inquiry by analyzing the approximability of several control actions under plurality, approval, and Condorcet voting, for both weighted and unweighted variants. We complete the resistance landscape outlined by Brelsford~\cite{bre:msthesis:approximations}, addressing previously unresolved cases.

\section{Conclusion}

We analyzed the approximability of the standard electoral control problems under plurality, approval, and Condorcet, in both the weighted and unweighted vote settings, resolving open cases left by prior work of Brelsford~\cite{bre:msthesis:approximations}. 

Our main contributions include new approximation algorithms, hardness/completeness results, and a refined classification of resistance across voting rules.
We also provide an axiomatic result for constructive control by deleting candidates. To our knowledge, every actively studied, ``natural'' voting rule satisfies that property.

For plurality, we give a strict reduction from the Minimum $k$-Union problem to constructive control by deleting candidates, thus giving a conditional lower bound on that control problem under the Dense vs. Random Conjecture~\cite{chl-din-mak:c:min-k-union}. For completeness results under approval, we reduce set covering problems into voter control problems and, in turn, model those voter control problems as covering integer programs (CIPs), leveraging known approximation results~\cite{kol-you:j:approx-cip}.
We show that all the remaining $\NP$-hard standard control problems have no approximation algorithms by giving types of arguments, depending on the nature of the control action and the voting rule.  For example, we show that partition-based control problems are inapproximable in general

We also provide an axiomatic result and prove that under any voiced voting rule, constructive control by candidate deletion lies in $\POLYAPX$. Our results deepen the understanding of approximability in voting and open new directions for proving unconditional hardness and extending techniques to broader voting systems. 

We find this work opens up many future directions. The natural first step is to extend this study to other voting rules, other control actions, and the nonunique-winner model. Moreover one could, for each of our partition problems whose decision problem is in $\P$, study the optimization version of that problem; in general, the $\P$ membership of a partition-based control problem does not necessarily imply a polynomial-time computable optimization counterpart~\cite{hem-hem-men:j:search-versus-decision}.  In a similar vein, one could explore other notions of approximations to attempt to circumvent the hardness results we present.

Another direction is to refine the lower bound for constructive control by deleting candidates under plurality, possibly involving a deeper understanding of the Minimum $k$-Union problem.  In fact, it would be interesting to find control problems that are equivalent (under some appropriate notion of equivalence/reduction) to $\MkU$. 
We would also find it interesting to explore axiomatic characterizations for approximability, e.g., by finding necessary and sufficient conditions on a voting rule $\cale$ for constructive control by deleting candidates under $\cale$ to be approximable. Our Lemma~\ref{lemma:ccdc-approx} provides a sufficient condition, but it is easy to see that one can design a voting rule $\cale$ that is not voiced and yet $\optd\cale\dash\CC\DC$ is approximable.
Perhaps more ambitious would be to develop dichotomy theorems for approximability in electoral control, aiming to classify when control problems---such as adding or deleting votes or candidates---are approximable or inapproximable. In fact, prior work has established dichotomies that separate $\P$ and $\NP$-hard cases of manipulation and bribery~\cite{fal-hem-hem:j:bribery, hem-hem:j:dichotomy-scoring,hem-sch:c:psrs}, and some results exist for control~\cite{hem-hem-sch:c:psrs,fal-hem-hem:j:weighted-control}.

\paragraph{Acknowledgement} This work was supported in part by  
the Laurie Bukovac Hodgson \& David Hodgson Endowed Fund, 
the Reid and Polly Anderson Endowment, and 
the Denison University Research Foundation. 
A two-page (excluding references) version of this paper appeared in the proceedings of the \emph{25th International Conference on Autonomous Agents and Multiagent Systems} (AAMAS 2026)~\cite{bui-cha-le-ngu:c-ea:approximating-control}, and a longer version of this work will appear in the proceedings of the \emph{9th International Conference on Algorithmic Decision Theory} (ADT~2026)~\cite{bui-cha-le-ngu:c:approximating-control}.
We thank the anonymous reviewers for their comments, which helped improve the presentation of this work. 

\bibliographystyle{alpha}
\newcommand{\etalchar}[1]{$^{#1}$}

\end{document}